\documentclass[12pt]{article}
\usepackage[margin=1in]{geometry}
\usepackage{setspace}
\usepackage{amsmath}
\usepackage{amsfonts}
\usepackage{amssymb}
\usepackage{amsthm}
\usepackage{graphicx}
\usepackage{mathtools}
\usepackage{sgame}
\usepackage{color}
\usepackage{natbib}
\usepackage{geometry}
\usepackage{comment}
\usepackage{caption}
\usepackage{subfigure}
\usepackage{float}
\usepackage{accents}

\usepackage{xparse}% http://ctan.org/pkg/xparse
\theoremstyle{plain}

\newtheorem{lem}{Lemma}

\newtheorem{prop}{Proposition}

\begin{document}

\title{Nash SIR: An Economic-Epidemiological Model of Strategic Behavior During a Viral Epidemic}
\author{David McAdams\thanks{Fuqua School of Business and Economics Department, Duke University, Durham, North Carolina, USA. Email: david.mcadams@duke.edu. I thank Carl Bergstrom, Yonatan Grad, and Marc Lipsitch for encouragement and Sam Brown, Troy Day, Nick Papageorge, Elena Quercioli, Lones Smith, Yangbo Song, Marta Wosinska, and participants at the Johns Hopkins Pandemic Seminar in April 2020 for helpful comments.
}
}

\date{May 3, 2020} 
\maketitle

\begin{abstract}
This paper develops a Nash-equilibrium extension of the classic SIR model of infectious-disease epidemiology (``Nash SIR''), endogenizing people's decisions whether to engage in economic activity during a viral epidemic and allowing for complementarity in social-economic activity. An \textit{equilibrium epidemic} is one in which Nash equilibrium behavior during the epidemic generates the epidemic. There may be multiple equilibrium epidemics, in which case the epidemic trajectory can be shaped through the coordination of expectations, in addition to other sorts of interventions such as stay-at-home orders and accelerated vaccine development. An algorithm is provided to compute all equilibrium epidemics.
%The resulting ``Behavioral SIR'' model can be used to anticipate the economic-epidemiological impacts of public-policy interventions (such as stay-at-home orders, expanded hospital capacity, and widespread testing of people who are not yet sick) and medical breakthroughs (such as a new treatment or vaccine) before and during a viral epidemic.
\end{abstract}

%%%%%%%%%%%%%%%%%%%

People's choices impact how a viral epidemic unfolds.
%, a notion widely appreciated by the infectious-disease epidemiologists tasked with projecting an epidemic's spread. 
As noted in a March 2020 \emph{Lancet} commentary on measures to control the current coronavirus pandemic, ``How individuals respond to advice on how best to prevent transmission will be as important as government actions, if not more important'' (Anderson et al (2020)).\nocite{Anderson_etal2010} Early on when pre-emptive measures could be especially effective (Dalton, Corbett, and Katelaris (2020)), \nocite{Dalton_etal2020} people are at little personal risk of exposure and hence may be unwilling to follow orders to ``distance'' themselves from others. On the other hand, as infections mount and the health-care system is overwhelmed, people may then voluntarily take extreme measures to limit their exposure to the virus. 
Clearly, the way in which people's incentives change during the course of an epidemic is essential to how the epidemic itself progresses, and how widespread are its harms.

This paper develops a Nash equilibrium extension of the classic Susceptible-Infected-Recovered (SIR) model of viral epidemiology. ``Nash SIR'' augments the well-known system of differential equations that characterizes epidemiological dynamics in the SIR model with a system of Bellman equations characterizing the dynamics of agent welfare and a Nash-equilibrium condition characterizing the dynamics of agent behavior. What emerges is a model of \emph{equilibrium epidemics} that, while highly stylized, sheds light on the interplay between epidemiological dynamics, economic behavior, and the health and economic harm done during the course of a viral epidemic. 

The paper's most important modeling innovation is to account for the \textit{economic complementarities} of personal interaction that can be lost when agents ``distance'' themselves to slow viral transmission. Such complementarities are missing from the existing literature (discussed below), but can impact the progression of an epidemic in meaningful ways.
In particular, a \textit{positive feedback} can arise in which people complying with public-health directives induces others to do so as well, and vice versa. As non-essential businesses close, there is less that people are able to do outside the home, reducing their incentive to go out. Similarly, as co-workers in an office (or professors in a university) stay home, there is less reason to go to the office yourself, especially when the work involved is collaborative and can be managed remotely.\footnote
{
The opposite is true of essential work. The more that essential workers are absent, the more valuable the work done by those who remain. More generally, there may be congestion effects associated with social-economic activity, increasing the benefit one gets as others reduce their activity. This paper abstracts from congestion for ease of exposition.   
}  

Developed independently, this paper's Nash-equilibrium SIR (``Nash SIR'') model generalizes the Nash SIR model in \cite{farboodi2020internal}, by allowing for complementarities in social-economic activity.  

In the traditional SIR model, the trajectory of the epidemic is completely determined by epidemiological fundamentals. Similarly, in \cite{farboodi2020internal}'s Nash SIR model, the epidemic has a unique equilibrium trajectory. By contrast, in this paper's Nash SIR model, there may be multiple potential trajectories for the epidemic, each of which induces agents to behave in a way that generates that epidemic trajectory. Because of this indeterminancy, the ultimate harm done during an epidemic, in terms of lost lives and lost livelihoods, can hinge on what agents believe about what others believe. This paper's model therefore highlights the importance of coordinating mechanisms, such as effective political leadership, in shaping expectations during an epidemic. 

In addition to \textit{coordinating interventions} such as a political leader's public statements, \textit{fundamental interventions} such as public policies, public-health programs, scientific effort, and new cultural practices impact the set of equilibrium-epidemic possibilities. Such impacts can be explored using the Nash SIR model, by computing the set of equilibrium epidemics with and without the intervention in question. To enable such exploration, I provide an algorithm to compute all equilibrium epidemics in any instance of the model. This algorithm requires solving a set of problems, each of which corresponds to a different potential ``final condition'' and involves solving a system of first-order differential equations that augments the traditional SIR equations. 

\paragraph{Relation to the literature.} This paper follows the dominant tradition within economics of modeling disease hosts as dynamically-optimizing agents with correct forward-looking beliefs. A few notable examples include \cite{geoffard1996rational}, \cite{kremer1996integrating}, \cite{adda2007behavior}, \cite{chan2016health}, and \cite{greenwood2019equilibrium}. More recently, there has been an outpouring of important work motivated by the SARS-CoV2 outbreak, much of it embedding economic models into an SIR (or closely related) framework. Some notable examples include \cite{alvarez2020simple}, \cite{bethune2020covid}, \cite{eichenbaum2020macroeconomics},  \cite{garabaldi2020modeling}, \cite{glover2020health}, \cite{jones2020optimal}, \cite{keppo2020avoidance}, \cite{krueger2020macroeconomic}, and \cite{toxvaerd2020equilibrium}.

There is also of course an enormous literature within epidemiology that models behavioral response to infectious disease. However, epidemiologists have been slow to adopt economics-style modeling, usually instead making ad hoc assumptions about behavior. For instance, \cite{bootsma2007effect} assume that people's intensity of social-contact avoidance during the 1918 flu pandemic varied depending on how many others in their community had recently died. An example that grapples with the dynamics of social distancing in the current pandemic is \cite{kissler2020social}. 

Thanks to the recent explosion of interest in economic epidemiology among economists, the gap between economics and infectious-disease epidemiology is closing. \cite{farboodi2020internal} provide an elegant Nash-equilibrium extension of the SIR model that augments the usual system of differential equations that governs epidemiological dynamics with just two additional differential equations. \cite{toxvaerd2020equilibrium} beautifully analyzes a similar equilibrium SIR model, establishing compelling features of the equilibrium trajectory.
Because of their analytical simplicity and tight connection to existing models and methods within epidemiology, \cite{farboodi2020internal}, \cite{toxvaerd2020equilibrium}, and others in this fast-growing literature could potentially have enormous influence on infectious-disease epidemiology, marrying the fields and promoting further cross-fertilization of ideas. 

Yet there is also a danger here. This new crop of equilibrium SIR models make an implicit assumption that the benefit people get from social-economic activity does not depend on others' activity. Consequently, the ``activity game'' that people play necessarily exhibits negative externalities (activity increases others' risk of infection) and strategic substitutes (increased risk of infection prompts others to be less active, \cite{bulow1985multimarket}). As a profession, we have strong insights about such games, insights that can be easily and powerfully communicated. These models could therefore be highly influential in terms of shaping public policy. However, the insights that we get from these models could be misguided if, in fact, the activity game exhibits positive externalities and/or strategic complements. This is especially important because, as I discuss in the concluding remarks, the qualitative nature of the game does indeed change during the course of the epidemic. 

The rest of the paper is organized as follows. Section \ref{section:model} presents the economic-epidemiological model, along with preliminary analysis. Section \ref{section:equilibrium} analyzes equilibrium epidemics in more detail. Section \ref{section:conclusion} discusses some limitations of the model and directions for future research.

%%%%%%%%%%%%%%%%%%%%%%%%%%%%%%%%%%%%%%
\section{Model and Preliminary Analysis}\label{section:model}
%%%%%%%%%%%%%%%%%%%%%%%%%%%%%%%%%%%%%%

This section presents the economic-epidemiological model, divided for clarity into three parts: \emph{the epidemic}, on how the epidemic process depends on agents' behavior (Section \ref{section:model:epidemic}); \emph{the economy,} on how the epidemic impacts economic activity, both directly by making people sick and indirectly by changing behavior (Section \ref{section:model:economy}); and \emph{individual and collective behavior,} on how the state of the epidemic and expectations about economic activity impact Nash-equilibrium behavior at each point along the epidemic trajectory (Section \ref{subsection:model:behavior}). 

%%%%%%%%%%%%%%%%%%%%%%%%%%%%%%%%%%%%%%
\subsection{The epidemic}\label{section:model:epidemic}

%\paragraph{Epidemiological states.}
There is a unit-mass population of hosts, referred to as ``agents.'' Building on the classic Susceptible-Infected-Recovered (SIR) model of viral epidemiology, each host is at each time $t \geq 0$ in one of five epidemiological states: ``susceptible'' ($S$) if as-yet-unexposed to the virus; ``carriage/contagious'' ($C$) if asymptomatically infected; ``infected/sick'' ($I$) if symptomatically infected; ``recovered from carriage'' ($R_C$) if immune but never sick; and ``recovered from sickness'' ($R_I$) if immune and previously sick.\footnote{It remains unknown whether those who recover from SARS-CoV2 infection are immune to re-infection and, if so, for how long (\cite{lipsitch2020immunity}. } 

\paragraph{Epidemiological distance.} 
At each point in time $t$, each agent who is not sick decides how intensively to distance themselves from others. Distancing with intensity $d_i \in [0,1]$ causes an agent to avoid fraction $\alpha d_i$ of ``meetings'' with other agents, where $\alpha \in (0,1]$ is a parameter capturing the maximal effectiveness of distancing. Agents who are sick are assumed to be automatically isolated, as if distancing with $\alpha= 1$. (The analysis extends easily to a more general context in which sick agents also transmit the virus.) 

Let $\Omega \equiv \{S, C, R_C, R_I\}$ denote the set of not-sick epidemiological states.  For each $\omega \in \Omega$, let $d_{\omega}(t)$ denote the average distancing intensity of those currently in state $\omega$ at time $t$ who choose to distance themselves. Let $\mathbf{d}_t \equiv \left( d_{\omega}(t'): \omega \in \Omega, 0 \leq t' < t \right)$ denote the \emph{collective distancing behavior} of the agent population up to time $t$, and let $\mathbf{d} \equiv \left( d_{\omega}(t): \omega \in \Omega, t \geq 0 \right)$ be their collective distancing behavior over the entire epidemic.

\paragraph{Epidemiological dynamics.}
The following notation is used to describe the state of the epidemic at each time $t \geq 0$, depending on agents' distancing behavior:
\begin{itemize}
    \item $S(t;\mathbf{d}_t) = $ mass of agents who are susceptible;
    \item $C(t;\mathbf{d}_t) = $ mass of agents who are in carriage, i.e., asymptotically  infected but not sick;
    \item $I(t;\mathbf{d}_t) = $ mass of agents who are sick;
    \item $R_C(;\mathbf{d}_t) = $ mass of agents who are immune and were not previously sick; and
    \item $R_I(t;\mathbf{d}_t) = $ mass of agents who are immune and were previously sick.
\end{itemize}
Because the population has unit mass, $\sum_{\omega \in \Omega} \omega(t) = 1$ for all $t$.

Agents transition between epidemiological states as follows:
\begin{itemize}
    \item $S \to C$: Susceptible agents become asymptomatically infected once ``exposed'' to someone currently infected, at a rate that depends on agents' behavior (details below);
    \item $C \to I$: Each agent with asymptomatic infection becomes sick at rate $\sigma > 0$; and
    \item $C \to R_C$ and $I \to R_I$: Each agent with infection clears their infection at rate $\gamma > 0$.
\end{itemize}
Initially at time $t=0$, mass $\Delta > 0$ have asymptomatic infection but no one is yet sick and no one is yet immune;\footnote
{The model can be easily extended to allow for innate immunity, by allowing some mass of hosts to be in states $R_I$ and $R_C$ at time $t=0$. For instance, during the ``second wave'' of SARS-CoV2 infections expected to arrive in Fall 2020, some hosts may retain immunity due to exposure during the first wave in Spring 2020. }
that is, $S(0) = 1 - \Delta$, $C(0) = \Delta$, and $I(0) = R_C(0) = R_I(0) = 0$.
Epidemiological dynamics at times $t > 0$ are then uniquely determined by the following system of differential equations:
\begin{align}
        S'(t;\mathbf{d}_t) & = - \beta (1 - \alpha d_S(t))(1 - \alpha d_C(t)) S(t;\mathbf{d}_t) C(t;\mathbf{d}_t)
\label{eqn_S}
\\
    C'(t;\mathbf{d}_t) & = - S'(t;\mathbf{d}_t) - (\sigma + \gamma) C(t;\mathbf{d}_t)
\label{eqn_C}
\\
    I'(t;\mathbf{d}_t) & = \sigma C(t;\mathbf{d}_t) - \gamma I(t;\mathbf{d}_t)
\label{eqn_I}
\\
    R_C'(t;\mathbf{d}_t) & = \gamma C(t;\mathbf{d}_t)
\label{eqn_RC}
\\
    R_I'(t;\mathbf{d}_t) & = \gamma I(t;\mathbf{d}_t)
\label{eqn_RI}
\end{align}
Let $\mathcal{E}(t;\mathbf{d}_t) \equiv  \left(S(t;\mathbf{d}_t), C(t;\mathbf{d}_t), I(t;\mathbf{d}_t), R_C(t;\mathbf{d}_t), R_I(t;\mathbf{d}_t) \right)$ denote the ``epidemic state'' at time $t$ and
$\mathcal{E}(\mathbf{d}) \equiv \left( \mathcal{E}(t;\mathbf{d}_t): t \geq 0 \right)$ the ``epidemic process.'' 

\vspace{.1in} \noindent \emph{Note on notation:}
I use ``$\mathbf{d}_t$ notation'' in equations (\ref{eqn_S}-\ref{eqn_RI}) to emphasize how the epidemic state at time $t$ depends on agents' previous distancing behavior. However, to ease exposition, I henceforth suppress this notation, except where needed for clarity.

\vspace{.1in} 
Equations (\ref{eqn_C}-\ref{eqn_RI}) are standard---reflecting agents' progression over time into carriage and then \emph{either} to infection at rate $\sigma$ \emph{or} to viral clearance at rate $\gamma$, and from infection to clearance at rate $\gamma$---but equation (\ref{eqn_S}) is different than in a standard SIR model. 

Each susceptible agent $i$ has a \emph{potential}  meeting (i.e., opportunity for transmission) with another randomly-selected agent $j$ at ``transmission rate'' $\beta > 0$. Since fraction $S(t)$ of the population is susceptible and fraction $C(t)$ have unisolated infection, the flow of potential meetings between susceptible and infected agents across the entire population is $\beta S(t)C(t)$. However, because  susceptible and contagious agents distance themselves with intensity $d_S(t)$ and $d_C(t)$, respectively, each such potential meeting is avoided with probability $(1-\alpha d_S(t))(1-\alpha d_C(t))$. The overall flow of newly-exposed hosts is therefore $\beta (1-\alpha d_S(t))(1-\alpha d_C(t))S(t) C(t)$, a functional form that appeared first in \cite{quercioli2006contagious}.

%%%%%%%%%%%%%%VACCINE/END OF EPIDEMIC%%%%%%%%%%%%
\paragraph{End of the epidemic.} For analytical convenience, I assume that the epidemic ends at time $T > 0$ when a vaccine is introduced, giving all still-susceptible agents subsequent immunity. (Infected agents remain infected, but there are no new infections after time $T$.) I focus on the case when $T < \infty$ and $T$ is known to all agents, but the analysis can be easily extended to a setting in which $T$ is a random variable drawn from interval support.

%%%%%%%%%%%%%%%%%%%%%%%%%%%%%%%%%%%%%%
\paragraph{Information states and distancing strategies.}
Agents' distancing decisions depend on what they know about their own epidemiological state and the overall epidemic. 
This paper focuses on the simplest non-trivial case, assuming that (i) agents know when they are sick but otherwise observe nothing about their own epidemiological state and (ii) agents have correct beliefs about the epidemic process. The model can be extended in several natural directions, to include diagnostic testing (allowing agents to learn more about their own epidemiological state) and incorrect beliefs, but such extensions are left for future work.  
%\footnote
%
%{The model can be extended to settings in which agents have limited information and/or incorrect beliefs about the epidemic. In that context, agents' uncertainty and/or incorrect beliefs are part of their information states, and efforts to educate the public can shape the course of the epidemic.}

Agent $i$'s \emph{information state} captures what she knows and believes, which depends only on (i) the time $t \geq 0$ and (ii) whether she is sick (state $I$), was previously sick (state $R_I$), or has not yet been sick (combined state $N \equiv S \cup C \cup R_C$). 

An agent currently in information state $\iota \in \{N,I,R_I\}$ is referred to as a ``$\iota$-agent.'' Let $N(t) = S(t) + C(t) + R_C(t)$ denote the mass of $N$-agents; thus, $N(t) + I(t) + R_I(t) = 1$.

Agent $i$'s \emph{distancing strategy} specifies her likelihood of distancing herself at each time $t$ in each information state. $I$-agents are automatically isolated, as mentioned earlier. $R_I$-agents know that they are immune and therefore have a dominant strategy not to distance themselves. It remains to determine the behavior of $N$-agents. 

Let $d_N(t)$ denote the share of $N$-agents who distance themselves. Because susceptible and contagious agents are in the same not-yet-sick information state, $d_N(t) = d_S(t) = d_C(t)$ and equation (\ref{eqn_S}) simplifies to:
\begin{equation}
        S'(t) = - \beta (1-\alpha d_N(t))^2 S(t) C(t)
\label{eqn_S_sym}
\end{equation}

\paragraph{Attack rate. }
Each agent's ex ante likelihood of becoming infected, referred to as the ``attack rate'' of the virus, is equal to $\lim_{t \to \infty} (R_C(t) + R_I(t))$. The attack rate is always strictly less than one, even if a vaccine is never discovered ($T = \infty$) and the epidemic is left completely uncontrolled; see \cite{brauer2012mathematical} and \cite{katriel2012attack} for details.

%%%%%%%%%%%%%%%%%%%%%%%%%%%%%%%%%%%%%%
\subsection{The economy}\label{section:model:economy}
%%%%%%%%%%%%%%%%%%%%%%%%%%%%%%%%%%%%%%

Each agent's activities fall into three broad categories: \emph{isolated activities} that can be performed while distancing (e.g., lifting weights, collaborating online), \emph{public activities} that require entering public spaces but do not require interacting with others (e.g., going for a walk, getting gas), and \emph{social activities} that require interacting physically with others (e.g., meeting friends, working in an office building). An agent who distances herself with intensity $d_i$ can continue engaging in isolated activity, but forgoes fraction $\alpha d_i$ of public and social activity and reduces others' opportunities to join her in social activity. %For instance, there is little to gain by going to work if everyone else is staying home.

\paragraph{Availability for social interaction.}
A not-sick agent who does not distance enjoys all the benefits of public activity, but engages in social activity only with those who are ``socially available.'' Let $A(t)$ denote agents' \textit{availability} for social interaction at time $t$, averaged across the entire population:
\begin{equation}\label{eqn:At}
    A(t) = (1-\alpha d_N(t)) N(t) + R_I(t)
        = 1 - I(t) - \alpha d_N(t) N(t).
\end{equation}
(Recall that $I$-agents are completely unavailable due to sickness, while $R_I$-agents find it optimal not to distance themselves at all.)

\paragraph{Economic output.}
Economic activity generates \textit{benefits}, a broad concept that should be understood to include everything from income (work activity) and access to goods and services (shopping) to psycho-social well-being (from interactions with friends). Sick agents are assumed for simplicity to be incapacitated and hence unable to engage in any economic activity; their economic benefit is normalized to zero. The benefit that well agents get depends on their own and others' distancing decisions, as well as how many people are currently sick.

Let $b(d_i; A)$ denote the flow benefit that agent $i$ gets when well and choosing distance $d_i \in \{0,1\}$, given population-wide average availability $0 \leq A \leq 1$. For concreteness, I assume that
\begin{equation}\label{eqn:bdef}
    b(d_i; A) = a_0 + a_1 (1-\alpha d_i) + a_2 (1 - \alpha d_i)A.
\end{equation}

\vspace{.1in} \noindent \textit{Discussion: meaning of the economic parameters.}
The parameters $a_0,a_1,a_2 > 0$ capture the importance, respectively, of isolated activities, public activities, and social activities for agent welfare. More precisely: $a_0$ captures the baseline level of benefits that a well agent gets while quarantined in the home; $a_1$ captures the extra benefits associated with being able to leave the home, e.g., the extra pleasure and health benefit of walking outside, the extra convenience of shopping in person rather than online; and $a_2$ captures the extra benefits associated with sharing the same physical space with others, e.g., eating out at a restaurant rather than at home, hugging a friend rather than just talking on the phone. (Put differently: $a_2$ is the cost associated with everyone else being quarantined; $a_1$ is the cost of quarantining yourself, in a world where everyone else is quarantined; and $a_0$ is the cost of being sick, in a world where everyone is quarantined.) These parameters can be changed in many ways. For instance, a restaurant service that delivers safely-prepared fresh-cooked meals would increase $a_0$ and reduce $a_2$, as would improved virtual-meeting technology that enhances remote collaboration. 

\vspace{.1in} \noindent \textit{Discussion: functional form of economic benefits.}
The assumption that the benefits of public and social activity are linear in own and others' availability simplifies the presentation but is not essential for the analysis or qualitative findings. For instance, suppose that agents were to prioritize their activities, e.g., by leaving the home only to get urgently-needed supplies, or to visit only with their dearest friends. In that case, each agent's benefit from public and social activity would naturally be a concave function of her own personal distance and of others' availability. The analysis can be easily adapted to allow for such concavity, but at the cost of complicating the presentation.

\paragraph{Economic losses due to the virus.}
If the virus did not exist, then no one would become sick and everyone would choose not to distance. All agents would then get constant flow economic benefit $b(0;1) = a_0 + a_1 + a_2$ and, since the population has unit mass, overall economic activity would also be $b(0;1)$.
The virus reduces economic activity directly, by making people sick, and indirectly, by inducing not-yet-sick agents to distance themselves. Distancing in turn creates two sorts of economic harm: ``private harm'' that distancing oneself reduces one's own public and social activity, and ``social harm'' that distancing oneself reduces others' social activity. 

Let $b_t(d_i)$ be shorthand for each well agent's flow economic benefit at time $t$. $b_t(d_i)$ depends on (i) how many people are recovered from sickness, $R_I(t)$, and how many are currently sick, $I(t) = 1 - N(t) - R_I(t)$, (ii) what fraction $d_N(t)$ of not-yet-sick agents are distancing, and (iii) her own distancing choice $d_i \in \{0,1\}$: 
\begin{align*}
    b_t(d_i) & = a_0 + a_1 (1 - \alpha d_i) + a_2 ( 1- \alpha d_i) A(t)   
\\
    & = a_0 + a_1 (1 - \alpha d_i) + a_2 ( 1- \alpha d_i) ((1-\alpha d_N(t)) N(t) + R_I(t))   
\end{align*}
All agents suffer economically throughout the epidemic, relative to the no-virus benchmark case in which everyone gets flow benefit $a_0 + a_1 + a_2$:
\begin{itemize}
    \item \textit{Sick:} $I$-agents are incapacitated and get zero economic benefit. These agents lose $a_0 + a_1 + a_2$.
    \item \textit{Previously sick:} $R_I$-agents do not distance, but have less opportunity for social interaction due to others' distancing behavior. These agents lose social-activity benefit $a_2(1-A(t))$.
    \item \textit{Not-yet-sick:} $N$-agents choose distancing intensity $d_N(t)$, reducing their public and social activities by a factor of $(1-\alpha d_N)$. These agents lose public-activity benefit $a_1 \alpha d_N$ and lose social-activity benefit $a_2 (1 - (1-\alpha d_N)A(t)) = a_2 (1 - A(t) + \alpha d_N A(t))$. 
\end{itemize}

Let $\Gamma_E(t)$ denote the lost economic activity at time $t$, across the entire population. Overall economic loss across the entire epidemic is $\Gamma_E = \int_0^{\infty} \Gamma_E(t) dt$. (If future losses are discounted by discount factor $0 < \delta \leq 1$, then the overall economic loss has present value $\Gamma_E = \int_0^{\infty} \delta^t  \Gamma_E(t) dt$ at time $0$. I focus on the case without discounting for ease of exposition.)

\begin{lem}\label{lem:1}
$\Gamma_E(t) = a_0 I(t) + a_1 (1-A(t)) + a_2 (1-A(t)^2)$. 
\end{lem}
\begin{proof}
See the Appendix.
\end{proof}

%%%%%%%%%%%%%%%%%%%%%%%%%%%%%%%%%%%%%%
\subsection{Individual welfare and equilibrium behavior}\label{subsection:model:behavior}
%%%%%%%%%%%%%%%%%%%%%%%%%%%%%%%%%%%%%%

Each agent seeks to minimize her own total\footnote
{The analysis can be trivially extended to allow for discounting of future losses.} losses during the course of the entire epidemic. 

Let $l_i(t)$ denote agent $i$'s \textit{flow loss}  at time $t$. As discussed earlier: $l_i(t) = a_0 + a_1 + a_2$ if $i$ is sick; $l_i(t) = a_2(1-A(t))$ if $i$ is well and not distancing, where $A(t)$ is others' availability for social interaction; and $l_i(t) = a_1 \alpha + a_2 (1 - A(t) + \alpha A(t))$ if $i$ is well and distancing.

Let $L_{\omega}(t)$ denote agent $i$'s expected future total losses starting from time $t$ if in epidemiological state $\omega \in \{S,C,I,R_C,R_I\}$, referred to as ``continuation losses from state $\omega$.'' (Continuation losses depend on future agent behavior and the future trajectory of the epidemic, but I suppress such notation as much as possible for ease of exposition.)
A susceptible agent who becomes infected at time $t$ will not notice this transition but, at that moment, her continuation losses change from $L_S(t)$ to $L_C(t)$. Let $H(t) \equiv L_C(t) - L_S(t)$ denote the ``harm of susceptible exposure'' at time $t$. 

Let $p_{i}(t)$ denote agent $i$'s subjective belief about her own likelihood of being susceptible at time $t$, conditional on being not-yet-sick.
At time $t$, mass $N(t)$ of agents are not-yet-sick, of whom mass $S(t)$ remain susceptible. Thus, $N$-agents' average likelihood of being susceptible is $\frac{S(t)}{N(t)}$. For simplicity, I will focus on epidemics with \textit{symmetric behavior} by all those in the same information state at each point in time, in which case $p_{i}(t) = \frac{S(t)}{N(t)}$.

\paragraph{Gain from distancing: reduced exposure.}
Suppose that, at time $t$, agent $i$ distances with intensity $d_i \in [0,1]$ and other $N$-agents distance themselves with intensity $d_N \in [0,1]$. 
Agent $i$ is then exposed to the virus at rate $\beta (1-\alpha d_i) (1-\alpha d_N)C(t)$, compared to being exposed at rate $\beta (1-\alpha d_N)C(t)$ if not distancing at all. The ``gain from distancing'' at time $t$, denoted $GAIN_t(d_N)$, is therefore 
\begin{equation}\label{eqn:GAIN}
    GAIN_t(d_i, d_N) = \alpha d_i \beta (1-\alpha d_N)C(t) \times H(t) \times \frac{S(t)}{N(t)}.
\end{equation}
The marginal gain from distancing $_t(d_N) = \frac{\mathrm{d} GAIN_t(d_i,d_N)}{\mathrm{d} d_i}$ is then 
\begin{equation}\label{eqn:MG}
    MG_t(d_N) = \alpha (1-\alpha d_N) \frac{\beta S(t)C(t)H(t)}{N(t)}.
\end{equation}
Note that the marginal gain from distancing is decreasing in $d_N$. This is because, as others distance themselves more, agents face less risk of exposure.

\paragraph{Economic cost of distancing: reduced activity.}
If other $N$-agents choose distancing intensity $d_N$, agent $i$ gets flow economic benefit $a_0 + a_1 (1-\alpha d_i) +  a_2 (1-\alpha d_i) ((1-\alpha d_N) N(t) + R_I(t))$ when choosing distancing intensity $d_i$, compared to $a_0 + a_1 + a_2 ((1-\alpha d_N) N(t) + R_I(t))$ when not distancing at all. The ``cost of distancing'' at time $t$, denoted $COST_t(d_i, d_N)$, is therefore 
\begin{equation}\label{eqn:COST}
    COST_t(d_i, d_N) = a_1 \alpha d_i + a_2 \alpha d_i ((1-\alpha d_N) N(t) + R_I(t)).
\end{equation}
The marginal cost of distancing $MC_t(d_N) = \frac{\mathrm{d} COST_t(d_i, d_N)}{\mathrm{d} d_i}$ is then 
\begin{equation}\label{eqn:MC}
    MC_t(d_N) = a_1 \alpha + a_2 \alpha ((1-\alpha d_N) N(t) + R_I(t)).
\end{equation}
Note that the marginal cost of distancing is decreasing in $d_N$. This is because, as others distance themselves more, there are fewer opportunities for social activity. 

Because the marginal gain and the marginal cost of distancing are each decreasing in $d_N$, the game that $N$-agents play may exhibit ``strategic substitutes'' or ``strategic complements'' (\cite{bulow1985multimarket}), depending on the epidemic state. By contrast, in \cite{quercioli2006contagious}, there are no sources of strategic complementarity. 

\paragraph{``Distancing game'' among agents.}
At each time $t$, the not-yet-sick $N$-agents play a \emph{game} when deciding whether or not to distance. (Sick $I$-agents are incapacitated, while previously sick $R_I$-agents obviously prefer not to distance. Thus, only $N$-agents have a non-trivial decision.) I assume that $N$-agents play a Nash equilibrium (NE) of this game, and focus on symmetric NE in which all $N$-agents choose the same distancing intensity. 

\begin{prop}\label{lem:unique}
Given epidemic state $\mathcal{E}(t) = (S(t), C(t), I(t), R_C(t), R_I(t))$ and harm from susceptible exposure $H(t) = L_C(t) - L_S(t)$, the ``time-$t$ distancing game'' played by not-yet-sick agents has a unique symmetric NE, in which agents choose distancing intensity $d_N^*(t)$. In particular: (i) if $MG_t(0) \leq MC_t(0)$, then $d_N^*(t) = 0$; (ii) if $MG_t(1) \geq MC_t(1)$, then $d_N^*(t) = 1$; and (iii) otherwise, if $MG_t(0) > MC_t(0)$ and $MG_t(1) < MC_t(1)$ then 
\begin{equation}
    d_N^*(t) = \frac{\frac{\beta S(t)C(t)H(t)}{N(t)} - a_1 - a_2(N(t) + R_I(t))}{\alpha \left( \frac{\beta S(t)C(t)H(t)}{N(t)} - a_2 N(t) \right)} \in (0,1).
\end{equation}
\end{prop}
\begin{proof}
The proof is in the Appendix.  
\end{proof}

\noindent Uniqueness of symmetric NE is not obvious, since the time-$t$ distancing game may have either strategic substitutes or strategic complements, depending on the epidemic state and the harm of susceptible exposure. However, uniqueness arises because $N$-agents have a dominant strategy whenever the game has strategic complements.

%%%%%%%%%%%%%%%%%%%%%%%%%%%%%%%%%%%%%%
\paragraph{Equilibrium epidemics.}  Let $\mathcal{E}(\mathbf{d}_N)$ denote the epidemic process that results when $N$-agents choose distancing intensity $d_N(t)$ at each time $t$, as determined by the system (\ref{eqn_C}-\ref{eqn_S_sym}). $\mathcal{E}^*$ is referred to as an \textit{equilibrium epidemic process} (or ``behaviorally-constrained epidemic'') if (i) $\mathcal{E}^* = \mathcal{E}(\mathbf{d}_N^*)$ and (ii) given the epidemic process $\mathcal{E}^*$, the time-$t$ distancing game has a symmetric NE in which $N$-agents choose distancing intensity $d_N^*(t)$, for all $t \geq 0$.

%%%%%%%%%%%%%%%%%%%%%%%%%%%%%%%%%%%%%%
\section{Equilibrium Epidemic Analysis} \label{section:equilibrium}
%%%%%%%%%%%%%%%%%%%%%%%%%%%%%%%%%%%%%%

This section characterizes all equilibrium epidemics with symmetric\footnote{I do not know if an equilibrium epidemic can exist with asymmetric behavior by symmetric agents. } agent behavior (or, more simply, ``equilibrium epidemics'') and provides an algorithm for computing them. 
The analysis is organized as follows. 
First, the augmented system of differential equations that governs economic-epidemiological dynamics in the Nash SIR model is presented. This system builds on the well-known system that governs epidemiological dynamics in the SIR model.
Second, for any given ``final prevalences'' $(S(T), C(T), I(T), R_I(T))$ at time $T$ when distancing ends (due to a perfect vaccine being introduced), there is at most one equilibrium epidemic having these final prevalences. %Moreover, an algorithm is provided to compute the entire equilibrium epidemic process, if it exists.

%I have not yet proven that an equilibrium epidemic exists, but see no reason to doubt existence in this context.

%%%%%%%%%%%%%%%%%%%%%%%%%%%%%%%%%%%%%%
\subsection{Economic-epidemiological dynamics} \label{subsection:dynamics}

At any given time $t$, the epidemic is characterized by: (i) the mass of agents in each epidemiological state ($S(t)$, $C(t)$, $I(t)$, $R_C(t)$, $R_I(t)$); (ii) the welfare of agents in each epidemiological state (as captured by state-contingent ``continuation losses'' $L_S(t)$, $L_C(t)$, $L_I(t)$, $L_{R_C}(t)$, $L_{R_I}(t)$); and (iii) the distancing behavior of agents who are not yet sick ($d_N^*(t)$).

\paragraph{Epidemiological dynamics.} 
The dynamics of the epidemic state $\mathcal{E}(t) = (S(t)$, $C(t)$, $I(t)$, $R_C(t)$, $R_I(t))$ up until time $T$ are determined by equations (\ref{eqn_C}-\ref{eqn_S_sym}) and depend on $N$-agents' distancing behavior. After the vaccine is introduced at time $T$, equations (\ref{eqn_C}-\ref{eqn_RI}) remain unchanged but, because there is no further transmission of the virus, $S'(t) = 0$.

\paragraph{Distancing behavior.} 
Lemma \ref{lem:unique} characterizes $N$-agents distancing behavior $d_N(t)$ at each time $t$, depending on the epidemic state $\mathcal{E}(t)$ and the harm of susceptible exposure $H(t) = L_C(t) - L_S(t)$. 

\paragraph{Welfare dynamics.} It remains to characterize how the continuation losses associated with each epidemiological state change over time.

\vspace{.1in} \noindent \emph{State $R_I$.}  
Agents who have recovered from sickness remain well\footnote{The analysis can be extended in a straightforward way to allow for the possibility of re-infection, for instance, by having recovered agents transition back at some rate to the susceptible state.} and choose not to distance. Such an agent still suffers from the fact that others are distancing, losing social-activity benefit $a_2(\alpha d_N^*(t)N(t) + I(t))$ relative to the no-virus benchmark in which everyone earns flow benefit $a_0 + a_1 + a_2$. Because these losses are ``sunk'' once experienced, and because $R_I$-agents do not transition to any other state, 
\begin{equation}\label{eqn:LRI}
    L_{R_I}'(t) = - a_2(\alpha d_N^*(t)N(t) + I(t)).
\end{equation}
After time $T$ when new transmission stops, all social distancing stops, i.e., $d_N^*(t) = 0$ for all $t > T$. However, well agents still suffer from not being able to engage socially with those who are sick. In particular, $L_{R_I}(t) = \int_{t' \geq t} a_2 I(t') \mathrm{d}t'$ for all $t \geq T$.

\vspace{.1in} \noindent \emph{State $I$.}
Sick agents incur flow loss $a_0 + a_1 + a_2$ and transition to the recovered state $R_I$ at rate $\gamma$. Thus,
\begin{equation}\label{eqn:LI}
    L_{I}'(t) = - (a_0 + a_1 +a_2) + \gamma (L_I(t) - L_{R_I}(t)).
\end{equation}

\vspace{.1in} \noindent \emph{State $R_C$.}  
Agents who have recovered from carriage never learn that they are immune, and so continue to distance themselves throughout the entire epidemic. In particular, these agents lose public-activity benefit $a_1 \alpha d_N^*(t)$, lose social-activity benefit $a_2 (1 - (1-\alpha d_N^*(t))((1-\alpha d_N^*(t))N(t) + R_I(t))$, and never transition to another state. Thus, 
\begin{equation}\label{eqn:LRC}
    L_{R_C}'(t) = - a_1 \alpha - a_2 (1 - (1-\alpha d_N^*(t))((1-\alpha d_N^*(t))N(t) + R_I(t)).
\end{equation}
After time $T$, because all social distancing stops and $R_C$-agents do not become sick, their only subsequent losses come from not being able to interact with other people who are sick, the same as $R_I$-agents. So, $L_{R_C}(t) = L_{R_I}(t)$ for all $t \geq T$.

\vspace{.1in} \noindent \emph{State $C$.}
Agents with asymptomatic infection incur the same flow losses due to social distancing as all not-yet-sick agents (including those in state $R_C$), but transition to sickness at rate $\sigma$ and to asymptomatic recovery at rate $\sigma$. Thus, 
\begin{equation}\label{eqn:LC}
    L_C'(t) = L_{R_C}'(t) + \gamma (L_I(t) - L_C(t)) + \sigma (L_{R_C}(t) - L_C(t)).
\end{equation}

\vspace{.1in} \noindent \emph{State $S$.}
Susceptible agents incur the same flow losses as all other not-yet-sick agents, but become asymptomatically infected at rate 
$\beta (1-\alpha d_N^*(t))^2 S(t)C(t)$. Thus, 
\begin{equation}\label{eqn:LS}
    L_S'(t) = L_{R_C}'(t) + \beta (1-\alpha d_N^*(t))^2 S(t)C(t) (L_C(t) - L_S(t)).
\end{equation}
After time $T$, $S$-agents remain susceptible and only suffer from not being able to interact with others who are sick, the same as $R_I$-agents. So, $L_S(t) = L_{R_I}(t)$ for all $t \geq T$.

%%%%%%%%%%%%%%%%%%%%%%%%%%%%%%%%%%%%%%
\subsection{Algorithm} \label{subsection:algorithm}
Suppose for a moment that an equilibrium exists with final epidemic state $\mathcal{E}(T)$. Here I discuss how to determine numerically whether an equilibrium epidemic exists with this ``final condition'' and, if so, to compute the entire epidemic trajectory. 

Observe first that the final epidemic state uniquely pins down the trajectory of the epidemic \textit{after} time $T$. Because there is no new transmission, no one distances and subsequent epidemiological dynamics are trivial: contagious agents leave state $C$ at rate $\gamma + \sigma$, fraction $\frac{\sigma}{\gamma + \sigma}$ becoming sick; sick agents recover at rate $\gamma$; and others remain in their current state. Moreover, because $C(T)$ and $I(T)$ together determine $(I(t): t \geq T)$, they also determine $L_{R_I}(t) = L_{R_C}(t) = L_S(t) = \int_{t' \geq t} a_2 I(t') \mathrm{d}t'$ for all $t \geq T$, which in turn determine $L_C(t)$ and $L_I(t)$ after $T$. 

Having determined $L_S(T)$ and $L_C(T)$, we now know $H(T) = L_C(T) - L_S(T)$, the harm of susceptible exposure just before the vaccine is introduced. Together with the final epidemic state, this uniquely determines $N$-agents' equilibrium distancing intensity just \textit{before} the vaccine is introduced, as characterized in Proposition \ref{lem:unique}. 

Having determined $N$-agent behavior $d_N^*(t)$, we now can determine: $S'(T)$ (equation (\ref{eqn_S_sym})) and all other epidemiological dynamics, which remain unchanged (equations (\ref{eqn_C}-\ref{eqn_RI})); $L_{R_I}'(T)$, which in turn determines $L_{I}'(T)$ (equations (\ref{eqn:LRI},\ref{eqn:LI})); and $L_{R_C}'(T)$, which together with $L_{I}'(T)$ determines $L_{C}'(T)$, which in turn determines $L_{S}'(t)$ (equations (\ref{eqn:LRC},\ref{eqn:LC},\ref{eqn:LS})).  
In this way, any \textit{candidate epidemic} can be uniquely traced backward over time, from the given final epidemic state (``final condition''), until one of three things happens: (i) the trajectory hits an invalid boundary\footnote{An ``invalid boundary'' is reached if $S(t)$, $C(t)$, $I(t)$, $R_C(t)$, or $R_I(t)$ equals zero at any time $t > 0$. }, in which case no equilibrium epidemic exists with the given final condition; (ii) the backwards trajectory ``ends'' at the desired initial epidemic state $\mathcal{E}(0) = (1 - \Delta, \Delta, 0,0,0)$, in which case a unique equilibrium  epidemic exists with the given final condition; or (iii) the backwards trajectory ends at some other initial epidemic state $\mathcal{E}(0) \neq (1 - \Delta, \Delta, 0,0,0)$, in which case no equilibrium epidemic exists with the given final condition.

%%%%%%%%%%%%%%%%%%%%%%%%%%%%%%%%%%%%%%
\section{Concluding Remarks}\label{section:conclusion}
%%%%%%%%%%%%%%%%%%%%%%%%%%%%%%%%%%%%%%

This paper introduces Nash SIR, an economic-epidemiological model of a viral epidemic that builds on the classic Susceptible-Infected-Recovered (SIR) model of infectious-disease epidemiology. The model departs from the previous literature by focusing on the complementarities associated with the social-economic activity that can be lost when agents distance themselves to prevent the spread of infection.

\paragraph{A changing game.}
An important complicating feature of this paper's model is that, as the epidemic progresses through its course, the basic strategic structure of the ``distancing game'' that agents play changes over time. For instance, very early in the epidemic when infection remains rare, the distancing game exhibits negative externalities, since agents get little health benefit but suffer substantial economic harm when others distance themselves. However, that changes once infection grows more common, as others' distancing generates greater health benefit. Moreover, the game can shift between having strategic substitutes and strategic complements.

\paragraph{Complementarity and multi-dimensionality of agent actions.}
This paper focuses on a simple context in which the only way to protect oneself from infection is to avoid public and in-person social activity. However, people can also prevent transmission in other ways, such as wearing a mask. Bearing that in mind, it would be interesting to generalize the analysis to allow agents to decide both (i) how much to curtail their public and social activities (``avoidance,'' as in this paper), and (ii) how much to change their behavior during such activities (``vigilance,'' as in \cite{quercioli2006contagious}). The game that agents play in this richer context has an interesting strategic structure, with agents' vigilance decisions always being strategic substitutes, agents' avoidance decisions potentially being either strategic complements or strategic substitutes, more vigilance promoting less avoidance, and more avoidance promoting less vigilance. 

\paragraph{Asymmetry and social inequality.} This paper assumes that agents are symmetric for ease of exposition, but this assumption appears to entail meaningful loss of generality. In particular, assuming that all agents are the same at the start of the epidemic obscures important issues related to inequality and social justice. To see why, suppose that agents belong to one of two social classes: ``elites'' who are able to earn income and care for themselves from home (higher $a_0$) and ``non-elites'' whose income and well-being hinge more on being in public social spaces (higher $a_2$). With less to lose by staying at home, elites will distance themselves relatively early during the epidemic. Having distanced less in the past, non-elites will then be more likely than elites to already have been exposed to the virus---further reducing their relative incentive to distance. In the end, the equilibrium trajectory of the epidemic could exacerbate pre-existing inequality, with non-elites bearing the brunt of the burden of the epidemic, being more likely to become sick and suffering more from the economic contraction associated with elite-driven distancing.

\pagebreak
\bibliographystyle{aer}
\bibliography{references}
\pagebreak
%\end{document}
%%%%%%%%%%%%%%%%%%%%%%%%%%%%%%%%%%%%%%%%%%%%%%%%%%%%%%
%%%%%%%%%%%%%%%%%%%%%%%%%%%%%%%%%%%%%%%%%%%%%%%%%%%%%%%%%%%

\appendix
\section{Mathematical proofs}

%%%%%%%%%%%%%%%%%%%%%%%%%%%%%%%%%%%%%%
\paragraph{Proof of Lemma \ref{lem:1}.}
\begin{proof}
Recall that $A(t) = (1-\alpha d_N(t))N(t) + R_I(t)$ and hence $1-A(t) = I(t) + \alpha d_N(t) N(t)$.

\vspace{.1in} \noindent \textit{Isolated activity:} Sick agents get no benefit, while well agents get full benefit $a_0$. The overall economic loss due to reduced isolated activity at time $t$ is therefore $a_0 I(t)$. 

\vspace{.1in} \noindent \textit{Public activity:} Sick agents get no benefit, well agents who do not distance get full benefit $a_1$, and well agents who distance get benefit $a_1 (1-\alpha)$. Since fraction $d_N(t)$ of $N$-agents distance and no $R_I$-agents distance, the overall economic loss due to reduced public activity at time $t$ is therefore $a_1 (I(t) + \alpha d_N(t) N(t)) = a_1(1-A(t))$.

\vspace{.1in} \noindent \textit{Social activity:} Sick agents get no benefit, well agents who do not distance get benefit $a_2 A(t)$, and well agents who distance get benefit $a_2(1-\alpha)A(t)$ (and hence lose $a_2(1 - A(t) +\alpha A(t))$). The overall economic loss due to reduced social activity at time $t$ is therefore $a_2$ times
\begin{align*}
    & I(t) + (1-A(t))(R_I(t) + (1-d_N(t)) N(t)) 
        + (1-A(t)+\alpha A(t))d_N(t)N(t)
\\
    & = {\color{red} I(t)} + (1-A(t))(R_I(t) + (1-d_N(t)) N(t)) 
        + ((1-A(t))(1-\alpha) + {\color{red}\alpha)d_N(t)N(t)}
\\
    & = {\color{red} 1 - A(t)} + (1-A(t)) \left( 
        {\color{blue}R_I(t) + (1-d_N(t))N(t) + (1-\alpha) d_N(t) N(t)}
    \right) 
\\
    & = (1-A(t)) \times \left( 1 +  
        {\color{blue}R_I(t) + (1-\alpha d_N(t))N(t)}
    \right) 
\\
    & = (1-A(t)) \times (1 +  
        {\color{blue}A(t)}) = 1-A(t)^2
\end{align*}
as desired.
\end{proof}

%%%%%%%%%%%%%%%%%%%%%%%%%%%%%%%%%%%%%%
\paragraph{Proof of Proposition \ref{lem:unique}.}
\begin{proof}
(i) \textit{No distancing:} If $MG_t(0) \leq MC_t(0)$, then the time-$t$ distancing game has a symmetric NE in which all agents choose not to distance, i.e., $d_N^*(t) = 0$. To establish uniqueness, note by equations (\ref{eqn:MG}-\ref{eqn:MC}) that $MG_t(0) \leq MC_t(0)$ implies $\frac{\beta S(t)C(t)H(t)}{N(t)} \leq a_1 + a_2(N(t) + R_I(t))$. But then
\begin{align*}
    MG_t(1) & = \alpha (1-\alpha) \frac{\beta S(t)C(t)H(t)}{N(t)}
    \\
    & \leq \alpha (1-\alpha) (a_1 + a_2(N(t) + R_I(t)))
    \\
    & < \alpha (a_1 + a_2((1-\alpha)N(t) + R_I(t)))
    \\
    & = MC_t(1)
\end{align*}
Since $MG_t(d_N)$ and $MC_t(d_N)$ are each linear in $d_N$, the fact that $MG_t(0) \leq MC_t(0)$ and $MG_t(1) < MC_t(1)$ implies that $MG_t(d_N) < MC_t(d_N)$ for all $d_N \in (0,1]$. In particular, $N$-agents have a dominant strategy not to distance. 

(ii) \textit{Maximal distancing:} If $MG_t(1) \geq MC_t(1)$, then a symmetric NE exists in which all agents choose to distance as much as possible, i.e., $d_N^*(t) = 1$. To establish uniqueness, note by equations (\ref{eqn:MG}-\ref{eqn:MC}) that $MG_t(1) \geq MC_t(1)$ implies $(1-\alpha)\frac{\beta S(t)C(t)H(t)}{N(t)} \geq a_1 + a_2((1-\alpha)N(t) + R_I(t))$. But then
\begin{align*}
    MG_t(0) & = \alpha \frac{\beta S(t)C(t)H(t)}{N(t)}
    \\
    & \geq \frac{\alpha}{1-\alpha} (a_1 + a_2((1-\alpha)N(t) + R_I(t)))
    \\
    & > \alpha (a_1 + a_2(N(t) + R_I(t)))
    \\
    & = MC_t(0)
\end{align*}
Since $MG_t(d_N)$ and $MC_t(d_N)$ are each linear in $d_N$, the fact that $MG_t(1) \geq MC_t(1)$ and $MG_t(0) > MC_t(0)$ implies that $MG_t(d_N) > MC_t(d_N)$ for all $d_N \in [0,1)$. In particular, $N$-agents have a dominant strategy to distance. 

(iii) \textit{Intermediate distancing:} If $MG_t(0) > MC_t(0)$ and $MG_t(1) < MC_t(1)$, then it must be that $MG_t'(d_N) = -  \frac{\alpha^2 \beta S(t)C(t)H(t)}{N(t)} < - \alpha^2 a_2 N(t) =  MC_t'(d_N)$ and hence that there exists a unique $d_N^*(t) \in (0,1)$ such that $MG_t(d_N^*(t)) = MC_t(d_N^*(t))$, $MG_t(d_N) > MC_t(d_N)$ for all $d_N < d_N^*(t)$, and $MG_t(d_N) < MC_t(d_N)$ for all $d_N > d_N^*(t)$. In particular, solving $MG_t(d_N^*(t)) = MC_t(d_N^*(t))$ yields
\begin{equation}
    d_N^*(t) = \frac{\frac{\beta S(t)C(t)H(t)}{N(t)} - a_1 - a_2(N(t) + R_I(t))}{\alpha \left( \frac{\beta S(t)C(t)H(t)}{N(t)} - a_2 N(t) \right)}.
\end{equation}
\end{proof}

%\end{document}
%%%%%%%%%%%%%%%%%%%%%%%%%%%%%%%%%%%%%%%%%%%%%%%%%%%%%%%%%%%%%%%%%%%%%%%%%%%%%%%%%%%%%%%%%%%%%%%%%%%%%%%%%%%%%%%%%%%%%%%%%%%%%%%%%%%%%%%%%%%%%%%%%%%%%%%%%%%%%%%%%%%%%%%%%%%%%%%%%%%%%%%%%%%%%%%%%%%%%%%%%%%%%%%%%%%%%%%%%%%%%%%%%%%%%%%%%%%%%%%%%%%%%%%%%%%%%%%%%%%%%%%%%

\end{document}